\documentclass[12pt]{article}
\usepackage{amsmath, amsfonts, amssymb, amsthm, graphicx, geometry, arydshln, umoline, subfig, newtxtext, newtxmath, color, comment, soul, enumerate, bm}
\usepackage{appendix}
\usepackage{tikz}
\usetikzlibrary{positioning}
\usepackage{natbib}
\usepackage[colorlinks=true,citecolor=blue]{hyperref}
\usepackage{booktabs}
\usepackage{makecell}
\usepackage{bbm}

\theoremstyle{definition}
\newtheorem{theorem}{Theorem}

\newtheorem{corollary}{Corollary}

\newtheorem{lemma}{Lemma}

\newtheorem{proposition}{Proposition}

\begin{document}

\title{Coalitional Manipulation and Target Set Correspondences}
\author{
Toshiya Yoshimura\thanks{Undergraduate School of Management, Department of Business Economics, Tokyo University of Science, 1-11-2, Fujimi, Chiyoda-ku, Tokyo, 102-0071, Japan. Email: tossyyossy01@gmail.com}
}
\date{\today}
\maketitle

\begin{abstract}
We study coalitional manipulation of interval-valued social choice correspondences on the single-peaked domain.
In the single-valued model, strategy-proofness and group strategy-proofness coincide within the standard class of peaks-only social choice functions, so generalized median voter rules retain their robustness to coalitional manipulation.
This equivalence breaks down once the outcome is allowed to be set-valued: generalized median correspondences may be strategy-proof while remaining vulnerable to coordinated deviations.
We show that, among generalized median correspondences allowing genuine set-valued choice, pairwise strategy-proofness already forces anonymity.
With voter sovereignty, this restriction becomes sharp and characterizes target set correspondences.
Thus, the target-set structure emerges as the form of set-valued flexibility compatible with pairwise coalitional robustness.

\end{abstract}

\section{Introduction}
Single-peaked preferences provide one of the classical environments in which strong positive results on strategy-proof voting are available.
Following the seminal work of \citet{black1948rationale}, a large literature has studied voting on a line when voters have single-peaked preferences.
A central result is due to \citet{Moulin1980}, who characterizes strategy-proof voting schemes in terms of generalized median voter rules; see also \citet{weymark2011unified}.
An important feature of these rules is their robustness to coordinated manipulation, reflecting the close relationship between individual and group strategy-proofness on the single-peaked domain; see also \citet{BarberaBergaMoreno2010}.

The situation becomes less straightforward when the social choice rule is allowed to select a set of alternatives.
\citet{klaus2020strategy} extend generalized median voting to interval-valued social choice correspondences and characterize generalized median correspondences under strategy-proofness and peaks-onliness.
Following their approach, and the decision-theoretic treatment of choice under complete uncertainty in \citet{Bossert2000uncertainty}, we evaluate set-valued outcomes through their best and worst alternatives.
Allowing interval-valued outcomes introduces flexibility that is absent in the single-valued model, since the final alternative may be selected from the chosen interval at a later stage.
At the same time, however, generalized median correspondences need not inherit the coalitional robustness of their single-valued counterparts.
This raises the question of how much set-valued flexibility can be maintained once robustness to coordinated manipulation is imposed.

We study this question by imposing robustness only against deviations by coalitions of at most two voters.
Related pairwise approaches to strategic manipulation are studied by \citet{Serizawa2006}, who considers effective pairwise strategy-proofness and self-enforcing manipulation.
Our pairwise requirement is substantially weaker than immunity to arbitrary coalitional deviations and provides a simple benchmark for robustness to coordinated manipulation.
Perhaps unexpectedly, even this limited requirement has a strong implication for set-valued median voting.
We show that, once genuine set-valued flexibility is allowed, pairwise strategy-proofness forces generalized median correspondences to be anonymous.
This sharply contrasts with the single-valued case, where generalized median voter rules may remain non-anonymous while satisfying strong coalitional incentive requirements.

Our second result connects this implication to the literature on target rules.
For single-valued choice, target point functions have been characterized through solidarity requirements such as population-monotonicity and replacement-dominance; see \citet{Thomson1993target}, \citet{ChingThomson1993}, \citet{Vohra1999}, and \citet{Klaus2001}. \citet{Klaus2020target} extend this approach to set-valued choice and characterize target set correspondences using population-monotonicity.
We provide a different foundation in a fixed-population setting: among genuinely set-valued social choice correspondences, voter sovereignty, strategy-proofness, and pairwise strategy-proofness characterize target set correspondences.
We also compare our main best--worst analysis with optimistic and pessimistic benchmark extensions.

The remainder of the paper is organized as follows.
Section~2 introduces the model, the best--worst extension, and the relevant incentive requirements.
Section~3 presents our main results on endogenous anonymity and the characterization of target set correspondences.
Section~4 discusses benchmark set extensions and related literature.
The proofs are collected in the Appendix.

\section{The Model}
\subsection{Preliminaries}
The set of voters is $N = \{1, \ldots, n\} \ (n \ge 3)$, and the set of alternatives is the interval of real numbers $A =[0,1] \subset \mathbb{R}$.
The preference of each voter $i \in N$ on $A$ is a complete, reflexive, and transitive binary relation $R_i$ on $A$.
Let $P_i$ denote the strict preference relations induced by $R_i$; that is, for all $x,y \in A$, $x P_i y$ if and only if $\lnot y R_i x$.
Similarly, let $I_i$ denote the indifference preference relations induced by $R_i$, namely, for all $x,y \in A$, $x I_i y$ if and only if $x R_i y  \land y R_i x$. 
We also assume that $R_i$ is single-peaked on $A$, i.e.,
each voter $i \in N$ has a unique alternative $\tau(R_i)$, such that for any $x,y \in A$,
\[
y < x\le \tau(R_i) \lor \tau(R_i) \le x < y \Rightarrow x P_i y, 
\]
where $\tau(R_i)$ is the top-ranked alternative in $A$; that is, $\tau(R_i) P_i x$ for all $x \in A \setminus \{\tau(R_i)\}$.
We denote the domain of single-peaked preferences by $\mathcal{SP}$, and let $\mathcal{SP}^n$ be the set of profiles $R := (R_i)_{i \in N}$ such that for each $i \in N, R_i \in \mathcal{SP}$.
Given \(R \in \mathcal{SP}^n\) and \(i \in N\), we write
\(R=(R_i,R_{-i})\), where \(R_i\) denotes agent \(i\)'s preference
relation and \(R_{-i}:=(R_j)_{j\in N\setminus\{i\}}\) denotes the
preference profile of all agents other than \(i\).
Similarly, given \(R \in \mathcal{SP}^n\) and \(C \subseteq N\), we write $R=(R_C,R_{-C}).$
For any $R \in \mathcal{SP}^n$, let $\tau(R)=( \tau(R_1), \ldots, \tau(R_n) )$.
Let $\underline \tau(R)=\min_{i \in N}\tau(R_i)$ and $\overline \tau(R)=\max_{i \in N}\tau(R_i)$.

\subsection{Standard Axioms}

Let $\mathcal{A}$ be the class of all nonempty closed intervals in $A$.
A \textit{social choice function} (SCF) is a function $f:\mathcal{SP}^n\to A$ that assigns a single alternative $f(R)\in A$ to every profile $R\in\mathcal{SP}^n$.
A \textit{social choice correspondence} (SCC) is a function $F:\mathcal{SP}^n\to\mathcal{A}$ that assigns a nonempty closed interval $F(R)\in\mathcal{A}$ to every profile $R\in\mathcal{SP}^n$.

We first introduce peaks-onliness for social choice functions and social choice correspondences.

\paragraph{Peaks-onliness for social choice functions}
An SCF $f:\mathcal{SP}^n\to A$ satisfies \textit{peaks-onliness} if, for every $R,R'\in\mathcal{SP}^n$ such that $\tau(R)=\tau(R')$, we have $f(R)=f(R')$.

\paragraph{Peaks-onliness for social choice correspondences}
An SCC $F:\mathcal{SP}^n\to\mathcal{A}$ satisfies \textit{peaks-onliness} if, for every $R,R'\in\mathcal{SP}^n$ such that $\tau(R)=\tau(R')$, we have $F(R)=F(R')$.

We next introduce two standard requirements concerning the treatment of voters.

\paragraph{Anonymity}
An SCC $F:\mathcal{SP}^n\to\mathcal{A}$ satisfies \textit{anonymity} if, for every $R\in\mathcal{SP}^n$ and every permutation $\sigma:N\to N$, we have
$F(R)=F(R^\sigma)$, where $R^\sigma:=(R_{\sigma(i)})_{i\in N}$.

Anonymity requires the social choice rule to treat voters symmetrically: permuting the identities of voters while leaving the profile of preferences otherwise unchanged does not affect the outcome.

We next impose a mild richness requirement on the range of the social choice rule.

\paragraph{Voter sovereignty}
An SCC $F:\mathcal{SP}^n\to\mathcal{A}$ satisfies \textit{voter sovereignty} if, for every $x\in A$, there exists $R\in\mathcal{SP}^n$ such that
$F(R)=\{x\}$.

Voter sovereignty requires every alternative to be attainable as the unique social choice at some preference profile.
Thus, it rules out social choice correspondences whose range excludes some alternatives.

\subsection{Best-Worst Set Extensions and Strategy-Proofness}

For each voter $i\in N$, we extend the preference relation $R_i$ over alternatives to a relation over intervals.
For each $X\in\mathcal{A}$, let $\underline{X}$ and $\overline{X}$ denote the minimum and maximum elements of $X$, respectively.
For each $R_i\in\mathcal{SP}$ and each $X\in\mathcal{A}$, define the set of best alternatives in $X$ by $b_{R_i}(X):=\{x\in X:xR_i y\text{ for every }y\in X\}$ and the set of worst alternatives in $X$ by $w_{R_i}(X):=\{x\in X:yR_i x\text{ for every }y\in X\}$.

By single-peakedness, $b_{R_i}(X)\subseteq\{\underline{X},\tau(R_i),\overline{X}\}$ and $|b_{R_i}(X)|=1$.
Moreover, $w_{R_i}(X)\subseteq\{\underline{X},\overline{X}\}$.
In particular, if $w_{R_i}(X)=\{\underline{X},\overline{X}\}$, then $\underline{X}\,I_i\,\overline{X}$.
Following \citet{klaus2020strategy}, with a slight abuse of notation, we treat $b_{R_i}(X)$ and $w_{R_i}(X)$ as alternatives whenever no confusion arises.

When a social choice correspondence selects an interval rather than a single alternative, incentive properties depend on how preferences over alternatives are extended to interval-valued outcomes.
We follow \citet{klaus2020strategy}, who study strategy-proofness for interval-valued social choice correspondences on the single-peaked domain using the best--worst extension.
This approach is also related to the decision-theoretic treatment of choice under complete uncertainty in \citet{Bossert2000uncertainty}, in
which the best and worst consequences of a feasible set play a central role.
Accordingly, we evaluate an interval through its best and worst alternatives.

\paragraph{Best--worst extension}
For any $X,Y\in\mathcal{A}$, $X R_i^{BW}Y$ if and only if $b_{R_i}(X)R_i b_{R_i}(Y)$ and $w_{R_i}(X)R_i w_{R_i}(Y)$.

The associated strict relation is defined by $X P_i^{BW}Y$ if and only if $X R_i^{BW}Y$ and either $b_{R_i}(X)P_i b_{R_i}(Y)$ or $w_{R_i}(X)P_i w_{R_i}(Y)$.

The best--worst extension need not be complete.

We first recall the standard incentive requirements for social choice functions.

\paragraph{Strategy-proofness for social choice functions}
A social choice function $f:\mathcal{SP}^n\to A$ is \textit{strategy-proof} if, for every $R\in\mathcal{SP}^n$, every $i\in N$, and every $R_i'\in\mathcal{SP}$, $f(R)R_i f(R_i',R_{-i})$.

\paragraph{Group strategy-proofness for social choice functions}
A social choice function $f:\mathcal{SP}^n\to A$ is \textit{group strategy-proof} if there do not exist $R=(R_C,R_{-C})\in\mathcal{SP}^n$, a nonempty coalition $C\subseteq N$, and $R_C'\in\mathcal{SP}^{|C|}$ such that $f(R_C',R_{-C})R_i f(R)$ for every $i\in C$ and $f(R_C',R_{-C})P_j f(R)$ for some $j\in C$.

Group strategy-proofness rules out coordinated deviations by any coalition in which every member weakly benefits and at least one member strictly benefits from the joint misreport.

We now turn to social choice correspondences.

\paragraph{Strategy-proofness}
Following \citet{klaus2020strategy}, an SCC $F:\mathcal{SP}^n\to\mathcal{A}$ is \textit{strategy-proof} if, for every $R\in\mathcal{SP}^n$, every $i\in N$, and every $R_i'\in\mathcal{SP}$, $F(R)R_i^{BW}F(R_i',R_{-i})$.

Since the best--worst extension is not necessarily complete, strategy-proofness requires the truthful outcome to be weakly preferred to, and hence comparable with, every outcome obtainable by a unilateral deviation.

For coalitional deviations, we adopt a weaker notion of profitable deviation.
A deviation is regarded as profitable if every member of the coalition weakly prefers the deviating outcome to the truthful outcome and at least one member strictly prefers it under the best--worst extension.
Unlike individual strategy-proofness defined above, this requirement does not require the truthful outcome to be weakly preferred to every outcome obtainable by a deviation.
In particular, a deviation leading to an outcome that is incomparable with the truthful outcome does not constitute a violation.
Thus, pairwise strategy-proofness, as defined below, does not in general imply individual strategy-proofness.

\paragraph{Pairwise strategy-proofness}
An SCC $F:\mathcal{SP}^n\to\mathcal{A}$ is \textit{pairwise strategy-proof} if there do not exist $R=(R_C,R_{-C})\in\mathcal{SP}^n$, a nonempty coalition $C\subseteq N$ with $|C|\leq2$, and $R_C'\in\mathcal{SP}^{|C|}$ such that $F(R_C',R_{-C})R_i^{BW}F(R)$ for every $i\in C$ and $F(R_C',R_{-C})P_j^{BW}F(R)$ for some $j\in C$.

\section{Results}

\subsection{Pairwise Strategy-Proofness and Anonymity}

We first recall the class of generalized median voter rules.

Let $\alpha=(\alpha_M)_{M\subseteq N}\in A^{2^n}$ satisfy $\alpha_L\geq\alpha_M$ for every $L\subseteq M\subseteq N$.
For each $R\in\mathcal{SP}^n$, let $\pi:N\to N$ be a permutation such that $\tau(R_{\pi(1)})\leq\cdots\leq\tau(R_{\pi(n)})$, and define
$\widetilde{\alpha}_{R}
:=
(\alpha_{\emptyset},
\alpha_{\{\pi(1)\}},
\alpha_{\{\pi(1),\pi(2)\}},
\ldots,
\alpha_N)$.

The \textit{generalized median voter rule} associated with $\alpha$, denoted by $f_G^\alpha$, is defined, for every $R\in\mathcal{SP}^n$, by $f_G^\alpha(R)
=
\operatorname{med}
\{\tau(R_1),\ldots,\tau(R_n),
\alpha_{\emptyset},
\alpha_{\{\pi(1)\}},
\ldots,
\alpha_N\}$.

\begin{proposition}[Moulin \citeyearpar{Moulin1980}]
\label{prop_moulin}
For any social choice function $f:\mathcal{SP}^n\to A$, the following statements are equivalent:
\begin{enumerate}
    \item $f$ satisfies strategy-proofness and peaks-onliness;
    \item $f$ satisfies group strategy-proofness and peaks-onliness;
    \item $f$ is a generalized median voter rule.
\end{enumerate}

\end{proposition}

Proposition~\ref{prop_moulin} highlights a distinctive feature of the single-valued single-peaked model.
Within the class of peaks-only social choice functions, individual strategy-proofness is already equivalent to group strategy-proofness.
Thus, generalized median voter rules do not lose their strategic robustness when the manipulation concept is strengthened from unilateral to coalitional deviations.
This equivalence provides a benchmark for our analysis.

\paragraph{Generalized median correspondences}
Let $\alpha=(\alpha_M)_{M\subseteq N}$ and $\beta=(\beta_M)_{M\subseteq N}$ be two vectors in $A^{2^n}$ such that $\alpha_M\leq\beta_M$ for every $M\subseteq N$, and $\alpha_L\geq\alpha_M$ and $\beta_L\geq\beta_M$ for every $L\subseteq M\subseteq N$.

For each $R\in\mathcal{SP}^n$, let $\pi:N\to N$ be a permutation such that $\tau(R_{\pi(1)})\leq\cdots\leq\tau(R_{\pi(n)})$, and define $\widetilde{\alpha}_{R}:=(\alpha_{\emptyset},\alpha_{\{\pi(1)\}},\alpha_{\{\pi(1),\pi(2)\}},\ldots,\alpha_N)$ and $\widetilde{\beta}_{R}:=(\beta_{\emptyset},\beta_{\{\pi(1)\}},\beta_{\{\pi(1),\pi(2)\}},\ldots,\beta_N)$.

For two vectors $x$ and $y$, let $\operatorname{med}(x,y)$ denote the median of all coordinates of $x$ and $y$.
The \textit{generalized median correspondence} associated with $(\alpha,\beta)$, denoted by $F_G^{\alpha,\beta}$, is defined, for every $R\in\mathcal{SP}^n$, by $F_G^{\alpha,\beta}(R):=[\operatorname{med}(\tau(R),\widetilde{\alpha}_{R}),\operatorname{med}(\tau(R),\widetilde{\beta}_{R})]$.

When some voters have the same peak, the value of $F_G^{\alpha,\beta}(R)$ does not depend on the particular permutation $\pi$ used to order them.

\begin{proposition}[\citet{klaus2020strategy}]
\label{prop_klaus}
For any SCC $F:\mathcal{SP}^n\to\mathcal{A}$, the following statements are equivalent:
\begin{enumerate}
    \item $F$ satisfies strategy-proofness and peaks-onliness;
    \item $F$ is a generalized median correspondence.
\end{enumerate}
\end{proposition}

Unlike in the single-valued case, strategy-proofness does not generally imply coalitional strategy-proofness for generalized median correspondences.
Indeed, \citet[Example~2]{klaus2020strategy} provide a generalized median correspondence that is strategy-proof but fails group strategy-proofness; their counterexample is generated by a coalition of only two voters.
Thus, the gap between individual and coalitional incentive compatibility already arises at the pairwise level.

We are interested in the strategic implications that arise specifically from set-valued choice.
As the preceding discussion shows, the gap between individual and coalitional strategy-proofness is a genuinely set-valued phenomenon.
If an SCC always selects a singleton, it can be identified with a social choice function, and the relevant incentive properties reduce to those
studied in Proposition~\ref{prop_moulin}.
We therefore exclude this single-valued case and impose the following minimal requirement.

\paragraph{Minimal flexibility}
An SCC $F:\mathcal{SP}^n\to\mathcal{A}$ satisfies \textit{minimal flexibility} if there exists $R\in\mathcal{SP}^n$ such that $\underline{F(R)}<\overline{F(R)}$.

Thus, minimal flexibility only requires that $F$ select a non-singleton interval at some preference profile.

Let $\mathcal{G}$ denote the class of all generalized median correspondences satisfying minimal flexibility.

Generalized median correspondences may treat voters asymmetrically, since the parameters $\alpha_M$ and $\beta_M$ may depend on the identity of the coalition $M$.
Thus, strategy-proofness and peaks-onliness alone do not guarantee equal treatment of voters.
Our first result shows that this asymmetry disappears once pairwise strategic robustness is imposed on minimally flexible generalized median correspondences.

\begin{theorem}
\label{thm_anonymity}
Let $F\in\mathcal{G}$.
If $F$ satisfies pairwise strategy-proofness, then $F$ is anonymous.
\end{theorem}

Theorem~\ref{thm_anonymity} is stronger than the anonymity conclusion alone suggests.
In fact, its proof shows that there exist $c,d\in A$ with $c<d$ such that $\alpha_M=c$ and $\beta_M=d$
for every nonempty coalition $M\subsetneq N$.
Thus, pairwise strategy-proofness eliminates not only non-anonymity of the social choice correspondence, but all dependence of the interior median parameters on the identities and composition of coalitions.
Only the two boundary parameters associated with $\emptyset$ and $N$ may remain unrestricted.

This restriction is particularly striking from the perspective of strategy-proof social choice.
Dictatorial rules provide a familiar example of strategy-proof rules that are maximally non-anonymous; indeed, in the single-valued single-peaked model, strong coalitional incentive compatibility is compatible with such asymmetry.
Proposition~\ref{prop_moulin} more generally allows non-anonymous generalized median voter rules to satisfy group strategy-proofness.
By contrast, once genuine set-valued flexibility is introduced, robustness against deviations by coalitions of only two voters already eliminates this asymmetry and makes anonymity an endogenous consequence
of incentive compatibility.

\subsection{Characterizing Target Set Correspondences}

The following proposition follows from Theorem~2 of \citet{klaus2020strategy}.

\begin{proposition}
\label{prop_sp_vs_po}
If an SCC $F:\mathcal{SP}^n\to\mathcal{A}$ satisfies strategy-proofness and voter sovereignty, then $F$ satisfies peaks-onliness.
\end{proposition}

Hence, by Propositions~\ref{prop_sp_vs_po} and~\ref{prop_klaus}, every strategy-proof SCC satisfying voter sovereignty is a generalized median correspondence.

\paragraph{Efficient median correspondences}
Let $a=(a_1,\ldots,a_{n-1})$ and $b=(b_1,\ldots,b_{n-1})$ be two vectors in $A^{n-1}$ such that $a_1\leq\cdots\leq a_{n-1}$, $b_1\leq\cdots\leq b_{n-1}$, and $a_k\leq b_k$ for every $k\in\{1,\ldots,n-1\}$.

The \textit{efficient median correspondence} associated with $(a,b)$, denoted by $F_M^{a,b}$, is defined, for every $R\in\mathcal{SP}^n$, by
$F_M^{a,b}(R):=[\operatorname{med}\{\tau(R_1),\ldots,\tau(R_n),a_1,\ldots,a_{n-1}\},\operatorname{med}\{\tau(R_1),\ldots,\tau(R_n),b_1,\ldots,b_{n-1}\}]$.

\begin{proposition}[\citet{klaus2020strategy}]
\label{prop_eff_median}
For any SCC $F:\mathcal{SP}^n\to\mathcal{A}$, the following statements are equivalent:
\begin{enumerate}
    \item $F$ satisfies strategy-proofness, voter sovereignty, and anonymity;
    \item $F$ is an efficient median correspondence.
\end{enumerate}
\end{proposition}

We next introduce the target set correspondence of \citet{Klaus2020target}, which generalizes the target rule studied by \citet{Thomson1993target}.

\paragraph{(Minimally flexible) Target set correspondences}
Fix a target interval $[c,d]\subseteq A$ with $c < d$.
The \textit{target set correspondence} associated with $[c,d]$, denoted by $F_T^{c,d}$, is defined by
\[
F_T^{c,d}(R)=
\begin{cases}
[c,d]\cap[\underline{\tau}(R),\overline{\tau}(R)] &
\text{if } [c,d]\cap[\underline{\tau}(R),\overline{\tau}(R)]\neq\emptyset,\\
\{\underline{\tau}(R)\} &
\text{if } \underline{\tau}(R)>d,\\
\{\overline{\tau}(R)\} &
\text{if } \overline{\tau}(R)<c.
\end{cases}
\]

The interval $[c,d]$ can be interpreted as a target region fixed in advance.
For instance, following the public-good interpretation of \citet{Klaus2020target}, it may represent a predetermined range of candidate locations for a public facility.
The correspondence retains this target region whenever it is compatible with efficiency, while adjusting it to the relevant efficient subset, or to the closest efficient point, when the voters' preferences make the entire target region inefficient.
Thus, the target set captures an ex ante institutional commitment while preserving set-valued discretion whenever such discretion is compatible with efficiency.

Each $F_T^{c,d}$ is an efficient anonymous median correspondence.
Indeed, $F_T^{c,d}$ is generated by the fixed-vote vectors $a=(c,\ldots,c)$ and $b=(d,\ldots,d)$.
Thus, every target set correspondence is an efficient median correspondence.

Let $\mathcal{T}$ denote the class of all target set correspondences.

We are now ready to state our main characterization result.

\begin{theorem}
\label{thm_target}
For any minimally flexible SCC $F:\mathcal{SP}^n\to\mathcal{A}$, the following statements are equivalent:
\begin{enumerate}
    \item $F$ satisfies strategy-proofness, voter sovereignty, and pairwise strategy-proofness;
    \item $F\in\mathcal{T}$.
\end{enumerate}
\end{theorem}

The converse direction establishes a stronger robustness property.
Every target set correspondence is immune to profitable deviations by arbitrary coalitions under the best--worst extension.
Hence, although only pairwise strategy-proofness is required in Theorem~\ref{thm_target}, the resulting class is in fact robust to coalitions of any size.

\section{Discussion}

\subsection{Benchmark Set Extensions}

We consider two benchmark extensions of preferences over intervals.

\paragraph{Optimistic and pessimistic extensions}
For any $X,Y\in\mathcal{A}$, $X R_i^{O}Y$ if and only if $b_{R_i}(X)R_i b_{R_i}(Y)$, and $X P_i^{O}Y$ if and only if $b_{R_i}(X)P_i b_{R_i}(Y)$.

Similarly, $X R_i^{P}Y$ if and only if $w_{R_i}(X)R_i w_{R_i}(Y)$, and $X P_i^{P}Y$ if and only if $w_{R_i}(X)P_i w_{R_i}(Y)$.

We refer to these as the optimistic and pessimistic extensions, respectively.

\paragraph{Pairwise O-strategy-proofness}
An SCC $F:\mathcal{SP}^n\to\mathcal{A}$ is \textit{pairwise O-strategy-proof} if there do not exist $R=(R_C,R_{-C})\in\mathcal{SP}^n$, a nonempty coalition $C\subseteq N$ with $|C|\leq2$, and $R_C'\in\mathcal{SP}^{|C|}$ such that $F(R_C',R_{-C})R_i^{O}F(R)$ for every $i\in C$ and $F(R_C',R_{-C})P_j^{O}F(R)$ for some $j\in C$.

\paragraph{Pairwise P-strategy-proofness}
An SCC $F:\mathcal{SP}^n\to\mathcal{A}$ is \textit{pairwise P-strategy-proof} if there do not exist $R=(R_C,R_{-C})\in\mathcal{SP}^n$, a nonempty coalition $C\subseteq N$ with $|C|\leq2$, and $R_C'\in\mathcal{SP}^{|C|}$ such that $F(R_C',R_{-C})R_i^{P}F(R)$ for every $i\in C$ and $F(R_C',R_{-C})P_j^{P}F(R)$ for some $j\in C$.

We first consider pairwise strategy-proofness under the pessimistic extension.
Under voter sovereignty, pessimistic strategic robustness is incompatible with minimal flexibility.
\begin{theorem}
\label{thm_gpsp}
There is no $F\in\mathcal{G}$ satisfying voter sovereignty
that is pairwise P-strategy-proof.
\end{theorem}

By contrast, the optimistic extension does not rule out set-valued flexibility.
The argument used in the proof of Theorem~\ref{thm_target} also shows that, within the class of voter-sovereign generalized median correspondences, pairwise O-strategy-proofness imposes the same target-set structure.
Thus, the target-set form remains the relevant restriction under the optimistic benchmark, although our main characterization result is stated
for the best--worst extension.

Thus, the contrast between the optimistic and pessimistic extensions is sharp.
Under the pessimistic extension, voter sovereignty and pairwise strategic robustness are incompatible with genuine set-valued flexibility.
Under the optimistic extension, by contrast, flexibility can be preserved, but only through the target-set structure.
Together with our main results under the best--worst extension, these observations show that the implications of coalitional robustness depend critically on how voters evaluate set-valued outcomes.

\subsection{Related Literature}
The target-set structure studied in this paper builds on a line of work on target rules in public-good allocation with single-peaked preferences.
For single-valued choice, target point functions have been characterized by solidarity requirements such as population-monotonicity and replacement-dominance.
In particular, \citet{ChingThomson1993} obtain a characterization based on population-monotonicity, while \citet{Thomson1993target} and \citet{Vohra1999} study target structures through replacement-dominance.
See also \citet{Klaus2001} and \citet{Gordon2007a} for extensions of these characterizations to broader public-good environments.

\citet{Klaus2020target} extend this approach to set-valued choice.
They show that, under the best--worst extension, efficiency and population-monotonicity characterize target set correspondences, while efficiency and replacement-dominance characterize the single-valued subclass of target point functions.
Thus, the existing characterizations of target structures are primarily based on solidarity requirements governing how outcomes should respond to changes in preferences or in the population.

Our characterization takes a different route.
We keep the population fixed and derive the target-set structure from incentive requirements.
In particular, strategy-proofness, voter sovereignty, and resistance to profitable deviations by coalitions of at most two voters are sufficient to characterize target set correspondences.
Hence, our result provides a strategic foundation for the target-set structure that is distinct from the solidarity-based foundations in the existing literature.

A related literature studies strategy-proofness for social choice correspondences under different extensions of preferences from alternatives to sets; see, for example, \citet{barbera2001strategy}, \citet{ching2002multi}, and \citet{IngalagaviSadhukhan2023}.
Our analysis focuses on the best--worst extension and, as benchmarks, the optimistic and pessimistic extensions.
The contrasting results in Section~4.1 show that coalitional robustness can depend substantially on how voters evaluate set-valued outcomes.
We leave the analysis under other set extensions for future research.

\subsection{Concluding Remarks}

This paper studies how coalitional incentive requirements restrict set-valued median voting on the single-peaked domain.
Our first result shows that, once genuine set-valued flexibility is required, pairwise strategy-proofness eliminates the asymmetries permitted by generalized median correspondences and endogenously yields anonymity.
With voter sovereignty, this restriction becomes even sharper and leads to the target-set structure.
From an institutional perspective, the target interval can be interpreted as a form of ex ante commitment that preserves some discretion in the final choice while limiting the scope for coordinated manipulation.

These results highlight a contrast with single-valued median voting, where strong coalitional incentive compatibility can coexist with non-anonymous generalized median voter rules.
Thus, with genuine set-valued flexibility, symmetry can emerge as a consequence of coalitional incentive compatibility rather than as an independent requirement.
The benchmark results under the optimistic and pessimistic extensions further show that the extent to which flexibility can be preserved depends on how voters evaluate set-valued outcomes.

\begin{center}
\Large{{\bf Appendix}}
\end{center}

\numberwithin{definition}{section}
\numberwithin{theorem}{section}
\numberwithin{lemma}{section}
\numberwithin{proposition}{section}
\numberwithin{corollary}{section}
\numberwithin{example}{section}
\renewcommand{\theequation}{\thesection.\arabic{equation}}
\appendix

\section{Omitted Proofs}

\subsection{Proof of Theorem \ref{thm_anonymity}}

\begin{lemma}
\label{lem_interior_flexibility}
Let $F=F_G^{\alpha,\beta}\in\mathcal{G}$ satisfy pairwise strategy-proofness.
Then there exists a nonempty coalition $K\subsetneq N$ such that
$\alpha_K<\beta_K$.
\end{lemma}

\begin{proof}
Suppose, to the contrary, that $\alpha_M=\beta_M$
for every nonempty coalition $M\subsetneq N$.

Since $F$ satisfies minimal flexibility, there exists some $H\subseteq N$ such that $\alpha_H<\beta_H$.
Indeed, if $\alpha_H=\beta_H$ for every $H\subseteq N$, then the two parameter vectors entering the lower and upper median operators coincide, and hence $F(R)$ is a singleton for every $R\in\mathcal{SP}^n$.
Therefore, under our supposition, either $\alpha_{\emptyset}<\beta_{\emptyset}$ or $\alpha_N<\beta_N$.
By symmetry, it suffices to consider the case
$\alpha_{\emptyset}<\beta_{\emptyset}$.

Suppose that $\alpha_{\emptyset}<\beta_{\emptyset}$.
Choose $t\in A$ such that $\alpha_{\emptyset}<t<\beta_{\emptyset}$, and fix two distinct voters $i,j\in N$.
Consider a profile $R\in\mathcal{SP}^n$ such that
$\tau(R_i)=t$ and $\tau(R_h)=1$ for every $h\neq i$.
Choose $R_i$ so that $\beta_{\emptyset}P_i\alpha_{\emptyset}$.
For every nonempty coalition $M\subsetneq N$, our supposition implies $\beta_M=\alpha_M\leq\alpha_{\emptyset}$.
Moreover, by monotonicity of $\beta$, we also have $\beta_N\leq\beta_M\leq\alpha_{\emptyset}$ for any nonempty coalition $M\subsetneq N$.
Hence, the median representation gives $F(R)=[\alpha_{\emptyset},t]$.

Now let $C=\{i,j\}$ and consider a report $R_C'$ such that $\tau(R_i')=1$ and $R_j'=R_j$.
Writing $R'=(R_C',R_{-C})$, all voters have peak $1$.
Therefore, the median representation gives
$F(R')=[\alpha_{\emptyset},\beta_{\emptyset}]$.
Since $\tau(R_j)=1$ and $t<\beta_{\emptyset}$, $\beta_{\emptyset} P_jt$.
Moreover, her worst alternative is $\alpha_{\emptyset}$ in both intervals.
Hence, $F(R')P_j^{BW}F(R)$.

For voter $i$, both intervals contain her peak $t$, so her best alternative is $t$ in both intervals.
Since $\beta_{\emptyset}P_i\alpha_{\emptyset}$, her worst alternative is $\alpha_{\emptyset}$ in both intervals.
Therefore, $F(R')R_i^{BW}F(R)$.
Thus, the coalition $C=\{i,j\}$ has a profitable deviation, contradicting pairwise strategy-proofness.

The case $\alpha_N<\beta_N$ is symmetric, which completes the proof.

\end{proof}

\begin{lemma}
\label{lem_local_rigidity}
Let $F=F_G^{\alpha,\beta}$ satisfy pairwise strategy-proofness.
Let $K\subsetneq N$ be a nonempty coalition such that $\alpha_K<\beta_K$.
For any nonempty coalition $M\subsetneq N$ such that
$|K\triangle M|=1$, we have
$\alpha_M=\alpha_K$ and $\beta_M=\beta_K$.
\end{lemma}

\begin{proof}
It suffices to consider the case in which $M\subset K$.
Let $K=M\cup\{j\}$ for some $j\in K$.

By monotonicity of the parameters,
$\alpha_M\geq\alpha_K$ and $\beta_M\geq\beta_K$.

We first show that $\alpha_M=\alpha_K$.
Suppose, to the contrary, that $\alpha_M>\alpha_K$.
Since $\alpha_K<\beta_K$, choose $x,t\in A$ such that
$\alpha_K<x<t<\min\{\alpha_M,\beta_K\}$.
Consider a profile $R\in\mathcal{SP}^n$ such that
$\tau(R_h)=x$ for every $h\in M$, $\tau(R_j)=t$, and
$\tau(R_h)=1$ for every $h\in N\setminus K$.
Choose $R_j$ so that $xP_j\beta_K$.

Let $m:=|M|$, and choose a permutation $\pi:N\to N$ such that $\{\pi(1),\ldots,\pi(m)\}=M$, $\pi(m+1)=j$, and $\{\pi(m+2),\ldots,\pi(n)\}=N\setminus K$.
For each $r\in\{0,\ldots,n\}$, let $S_r:=\{\pi(1),\ldots,\pi(r)\}$.
Then $S_m=M$ and $S_{m+1}=K$.

We first consider the lower endpoint.
For every $r\leq m$, since $S_r\subseteq M$, monotonicity implies $\alpha_{S_r}\geq\alpha_M>t$.
Hence, $|\{r\in\{0,\ldots,n\}:\alpha_{S_r}\geq t\}|\geq m+1$.
Moreover, $|\{h\in N:\tau(R_h)\geq t\}|=n-m$.
Therefore, $|\{h\in N:\tau(R_h)\geq t\}|+
|\{r\in\{0,\ldots,n\}:\alpha_{S_r}\geq t\}|
\geq n+1$.

On the other hand, for every $r \geq m+1$, since $K\subseteq S_r$, monotonicity implies
$\alpha_{S_r}\leq\alpha_K<x<t$.
Hence, $|\{r\in\{0,\ldots,n\}:\alpha_{S_r}\leq t\}|\geq n-m$.
Also, $|\{h\in N:\tau(R_h)\leq t\}|=m+1$.
Therefore,
$|\{h\in N:\tau(R_h)\leq t\}|+
|\{r\in\{0,\ldots,n\}:\alpha_{S_r}\leq t\}|
\geq n+1$.
Thus, $\operatorname{med}(\tau(R),\widetilde{\alpha}_R)=t$.

For the upper endpoint, for every $r\geq m+1$, we have
$\beta_{S_r}\leq\beta_K$.
Hence,
$|\{h\in N:\tau(R_h)\leq\beta_K\}|+
|\{r\in\{0,\ldots,n\}:\beta_{S_r}\leq\beta_K\}|
\geq (m+1)+(n-m)=n+1$.

Similarly, for every $r\leq m+1$, we have
$\beta_{S_r}\geq\beta_K$.
Hence,
$|\{h\in N:\tau(R_h)\geq\beta_K\}|+
|\{r\in\{0,\ldots,n\}:\beta_{S_r}\geq\beta_K\}|
\geq (n-m-1)+(m+2)=n+1$.
Thus,
$\operatorname{med}(\tau(R),\widetilde{\beta}_R)=\beta_K$,
and therefore $F(R)=[t,\beta_K]$.

Fix some $i\in M$ and let $C=\{i,j\}$.
Consider $R_C'$ such that $R_i'=R_i$ and $\tau(R_j')=x$, and write $R'=(R_C',R_{-C})$.
At $R'$, we have $\tau(R_h')=x$ for every $h\in K$ and
$\tau(R_h')=1$ for every $h\in N\setminus K$.

For the lower endpoint, $|\{h\in N:\tau(R_h')\leq x\}|=m+1$.
Moreover, for every $r\geq m+1$, since $K\subseteq S_r$, $\alpha_{S_r}\leq\alpha_K<x$.
Hence, $|\{r\in\{0,\ldots,n\}:\alpha_{S_r}\leq x\}|\geq n-m$.
Therefore, $|\{h\in N:\tau(R_h')\leq x\}|+
|\{r\in\{0,\ldots,n\}:\alpha_{S_r}\leq x\}|
\geq n+1$.
On the other hand, $|\{h\in N:\tau(R_h')\geq x\}|=n$, and $\alpha_{\emptyset}\geq\alpha_M>t>x$.
Hence, $|\{h\in N:\tau(R_h')\geq x\}|+ |\{r\in\{0,\ldots,n\}:\alpha_{S_r}\geq x\}| \geq n+1$.
Thus, $\operatorname{med}(\tau(R'),\widetilde{\alpha}_{R'})=x$.

For the upper endpoint, $|\{h\in N:\tau(R_h')\leq\beta_K\}|=m+1$ and $|\{r\in\{0,\ldots,n\}:\beta_{S_r}\leq\beta_K\}|\geq n-m$.
Similarly, $|\{h\in N:\tau(R_h')\geq\beta_K\}|=n-m-1$
and $|\{r\in\{0,\ldots,n\}:\beta_{S_r}\geq\beta_K\}|\geq m+2$.
Hence, $\operatorname{med}(\tau(R'),\widetilde{\beta}_{R'})=\beta_K$.
Therefore, $F(R')=[x,\beta_K]$.

Since $\tau(R_i)=x$, we have $b_{R_i}(F(R))=t$ and $b_{R_i}(F(R'))=x$.
Moreover, $w_{R_i}(F(R))=w_{R_i}(F(R'))=\beta_K$.
Since $xP_i t$, it follows that $F(R')P_i^{BW}F(R)$.
For voter $j$, since $\tau(R_j)=t$ and $t\in F(R)\cap F(R')$, we have $b_{R_j}(F(R))=b_{R_j}(F(R'))=t$.
Moreover, since $xP_j\beta_K$, we have $w_{R_j}(F(R))=w_{R_j}(F(R'))=\beta_K$.
Hence, $F(R')R_j^{BW}F(R)$.

Thus, the coalition $C=\{i,j\}$ has a profitable deviation, contradicting pairwise strategy-proofness.
Therefore, $\alpha_M=\alpha_K$.

By a symmetric argument, $\beta_M=\beta_K$.
Therefore, $\alpha_M=\alpha_K$ and $\beta_M=\beta_K$.
The case in which $K\subset M$ is symmetric.
Hence the result follows.
\end{proof}

\begin{proof}[Proof of Theorem~\ref{thm_anonymity}]
Let $F=F_G^{\alpha,\beta}\in\mathcal{G}$ satisfy pairwise strategy-proofness.

It is enough to show that there exist $c,d\in A$ such that $\alpha_M=c$ and $\beta_M=d$
for every nonempty coalition $M\subsetneq N$.
Indeed, in that case,
$\widetilde{\alpha}_R=(\alpha_{\emptyset},c,\ldots,c,\alpha_N)$
and
$\widetilde{\beta}_R=(\beta_{\emptyset},d,\ldots,d,\beta_N)$ for every $R\in\mathcal{SP}^n$ and every ordering of voters by their peaks.
Hence, $F$ is anonymous.

We now establish this equality of the parameters.
By Lemma~\ref{lem_interior_flexibility}, there exists a nonempty coalition $K\subsetneq N$ such that $\alpha_K<\beta_K$.
Let $c:=\alpha_K$ and $d:=\beta_K$.

Fix any nonempty coalition $M\subsetneq N$.
There exists a finite sequence of nonempty coalitions $K=S_0,S_1,\ldots,S_q=M$, with $S_r\subsetneq N$ for every $r$, such that $|S_r\triangle S_{r+1}|=1$ for every $r\in\{0,\ldots,q-1\}$.
Since $\alpha_{S_0}=c<d=\beta_{S_0}$, Lemma~\ref{lem_local_rigidity} implies $\alpha_{S_1}=c$ and $\beta_{S_1}=d$.
Thus, $\alpha_{S_1}<\beta_{S_1}$, and the lemma can be applied again.
Proceeding inductively, we obtain $\alpha_{S_r}=c$ and $\beta_{S_r}=d$ for every $r\in\{0,\ldots,q\}$.

In particular, $\alpha_M=c$ and $\beta_M=d$.
Since $M$ was arbitrary, this holds for every nonempty coalition $M\subsetneq N$.
Therefore, $\widetilde{\alpha}_R=(\alpha_{\emptyset},c,\ldots,c,\alpha_N)$ and
$\widetilde{\beta}_R=(\beta_{\emptyset},d,\ldots,d,\beta_N)$ for every $R\in\mathcal{SP}^n$ and every ordering of voters by their peaks.
Hence, $F$ is anonymous.
\end{proof}

The following corollary is derived directly from the proof of Theorem~\ref{thm_anonymity}.
\begin{corollary}
\label{cor_constant_parameters}
Let $F=F_G^{\alpha,\beta}\in\mathcal{G}$ satisfy pairwise strategy-proofness.
Then there exist $c,d\in A$ with $c<d$ such that
$\alpha_M=c$ and $\beta_M=d$
for every nonempty coalition $M\subsetneq N$.
\end{corollary}

\subsection{Proof of Theorem~\ref{thm_target}}
\begin{proof}
Suppose first that $F$ satisfies strategy-proofness, voter sovereignty, and pairwise strategy-proofness.

By Proposition~\ref{prop_sp_vs_po}, $F$ satisfies peaks-onliness.
Hence, by Proposition~\ref{prop_klaus}, $F$ is a generalized median correspondence.
Since $F$ also satisfies minimal flexibility, $F\in\mathcal{G}$.

By Corollary~\ref{cor_constant_parameters}, there exist $c,d\in A$ with $c<d$ such that $\alpha_M=c$ and $\beta_M=d$ for every nonempty coalition $M\subsetneq N$.
Voter sovereignty also determines the boundary parameters.
To attain the singleton $\{1\}$, the lower median must equal $1$ at some profile.
Since there are only $n$ voter peaks among the $2n+1$ arguments of the median, at least one $\alpha$-parameter must equal $1$.
By monotonicity, $\alpha_{\emptyset}$ is the largest $\alpha$-parameter, and hence $\alpha_{\emptyset}=1$.
Since $\alpha_{\emptyset}\leq\beta_{\emptyset}\leq1$, it follows that $\beta_{\emptyset}=1$.
Similarly, to attain the singleton $\{0\}$, at least one $\beta$-parameter must equal $0$.
Since $\beta_N$ is the smallest $\beta$-parameter, $\beta_N=0$.
Thus, $0\leq\alpha_N\leq\beta_N$ implies $\alpha_N=0$.

Therefore, for every profile $R$ and every ordering of voters by their peaks, $\widetilde{\alpha}_R=(1,c,\ldots,c,0)$ and $\widetilde{\beta}_R=(1,d,\ldots,d,0)$.
Hence, by the definition of a target set correspondence,
$F=F_T^{c,d}\in\mathcal{T}$.

Conversely, let $F=F_T^{c,d}\in\mathcal{T}$.
Since $F$ is an efficient median correspondence, Proposition~\ref{prop_eff_median} implies that $F$ satisfies strategy-proofness and voter sovereignty.
It remains to show pairwise strategy-proofness.

Let $C\subseteq N$ be a nonempty coalition and let $R'=(R_C',R_{-C})$. Write $F(R)=[l,u]$ and $F(R')=[l',u']$.
We first show that there cannot be a deviation such that
$b_{R_i}(F(R'))R_i b_{R_i}(F(R))$ for every $i\in C$ and
$b_{R_j}(F(R'))P_j b_{R_j}(F(R))$ for some $j\in C$.

Suppose such a deviation exists.
If $\tau(R_j)\in[l,u]$, then $b_{R_j}(F(R))=\tau(R_j)$,
so voter $j$ cannot strictly improve.
Hence, either $\tau(R_j)<l$ or $\tau(R_j)>u$.
By symmetry, suppose that $\tau(R_j)<l$.
Then a strict improvement requires $l'<l$.

If $l=c$, then $l'<c$ implies $\overline{\tau}(R')<c$.
Since $l=c$, we have $\overline{\tau}(R)\geq c$.
Choose $k\in N$ such that $\tau(R_k)\geq c$.
Since $\overline{\tau}(R')<c$, we must have $k\in C$.
Moreover, $b_{R_k}(F(R))\geq c$ whereas every alternative in $F(R')$ is strictly below $c$.
Thus, $b_{R_k}(F(R))P_k b_{R_k}(F(R'))$, a contradiction.

If $l<c$, then by the definition of a target set correspondence, $l=\overline{\tau}(R)$ and $F(R)=\{l\}$.
Since $l'<l<c$, we have $\overline{\tau}(R')<l$.
Choose $k\in N$ such that $\tau(R_k)=l$.
Then $k\in C$, and every alternative in $F(R')$ is strictly below $\tau(R_k)=l$.
Hence, $b_{R_k}(F(R))P_k b_{R_k}(F(R'))$, again a contradiction.
Thus, no coalition can weakly improve every member's best alternative and strictly improve some member's best alternative.

We next show that if $b_{R_i}(F(R'))R_i b_{R_i}(F(R))$
for every $i\in C$, then $F(R)\subseteq F(R')$.
It is enough to show that $l'\leq l$ and $u'\geq u$.

Suppose, to the contrary, that $l'>l$.
If $l<c$, then $l=\overline{\tau}(R)$ and $F(R)=\{l\}$.
Every voter has a peak weakly below $l$, whereas every alternative in $F(R')$ is strictly above $l$.
Hence, $b_{R_i}(F(R))P_i b_{R_i}(F(R'))$
for every $i\in C$, a contradiction.

If $l=c$, then $l'>c$ implies $\underline{\tau}(R')>c$.
Choose $k\in N$ such that $\tau(R_k)\leq c$.
Such a voter exists because $l=c$.
Moreover, $k\in C$, since otherwise $\underline{\tau}(R')\leq\tau(R_k)\leq c$.
Since every alternative in $F(R')$ is strictly above $c$, $b_{R_k}(F(R))P_k b_{R_k}(F(R'))$,
a contradiction.

Finally, if $l>c$, then $l=\underline{\tau}(R)$.
The inequality $l'>l$ implies $\underline{\tau}(R')>l$.
Choose $k\in N$ such that $\tau(R_k)=l$.
Then $k\in C$, and every alternative in $F(R')$ is strictly above $\tau(R_k)$.
Hence, $b_{R_k}(F(R))P_k b_{R_k}(F(R'))$, again a contradiction.

Therefore, $l'\leq l$.
By a symmetric argument, $u'\geq u$.
Hence, $F(R)\subseteq F(R')$.

Now suppose that $C$ has a profitable deviation under the best--worst extension.
Then $b_{R_i}(F(R'))R_i b_{R_i}(F(R))$ and $w_{R_i}(F(R'))R_i w_{R_i}(F(R))$ for every $i\in C$.

By the argument above, $F(R)\subseteq F(R')$.
Therefore, $w_{R_i}(F(R))R_i w_{R_i}(F(R'))$ for every $i\in C$.
Together with $w_{R_i}(F(R'))R_i w_{R_i}(F(R))$,
this implies $w_{R_i}(F(R'))I_i w_{R_i}(F(R))$ for every $i\in C$.

Hence, if some member strictly prefers $F(R')$ to $F(R)$ under the best--worst extension, the strict improvement must come from her best alternative.
Thus, there exists $j\in C$ such that
$b_{R_j}(F(R'))P_j b_{R_j}(F(R))$, while $b_{R_i}(F(R'))R_i b_{R_i}(F(R))$ for every $i\in C$.
This contradicts the first part of the argument.

Therefore, $F$ admits no profitable deviation by any nonempty coalition under the best--worst extension, and hence it is pairwise strategy-proof.
\end{proof}

\subsection{Proof of Theorem~\ref{thm_gpsp}}
\begin{proof}
Let $F\in\mathcal{G}$ satisfy voter sovereignty.
As shown in the proof of Theorem~\ref{thm_target}, voter sovereignty implies $\alpha_{\emptyset}=\beta_{\emptyset}=1$ and $\alpha_N=\beta_N=0$.

Since $F$ is minimally flexible, choose a nonempty coalition $M\subsetneq N$ such that $\alpha_M<\beta_M$ and $\alpha_H=\beta_H$ for every $H\supsetneq M$.

Let $m:=|M|$.
Consider a profile $R\in\mathcal{SP}^n$ such that $\tau(R_h)=0$ for every $h\in M$ and $\tau(R_h)=1$ for every $h\in N\setminus M$.
Then $F(R)=[\alpha_M,\beta_M]$.

Fix $i\in M$ and $j\in N\setminus M$.
Choose $z\in A$ such that $\alpha_M<z<\beta_M$.
Let $R_i'$ and $R_j'$ be preferences such that $\tau(R_i')=\tau(R_j')=z$, and write
$R':=(R_i',R_j',R_{-\{i,j\}})$.
At $R'$, we have $\tau(R'_h)=0$ for every $h\in M\setminus\{i\}$, $\tau(R'_i)=\tau(R'_j)=z$, and
$\tau(R'_h)=1$ for every $h\in N\setminus(M\cup\{j\})$.

Choose a permutation $\pi:N\to N$ such that $\{\pi(1),\ldots,\pi(m-1)\}=M\setminus\{i\}$, $\pi(m)=i$, $\pi(m+1)=j$, and $\{\pi(m+2),\ldots,\pi(n)\} = N\setminus(M\cup\{j\})$.
For each $r\in\{0,\ldots,n\}$, let $S_r:=\{\pi(1),\ldots,\pi(r)\}$.
Then $S_m=M$ and $S_{m+1}=M\cup\{j\}$.

Write $F(R')=[l',u']$.
For every $r\leq m$, we have $S_r\subseteq M$, and hence
$\alpha_{S_r} \geq \alpha_M$ by monotonicity.
Therefore,
$|\{r\in\{0,\ldots,n\}:\alpha_{S_r}<\alpha_M\}| \leq n-m$.
Moreover, $|\{h\in N:\tau(R_h')<\alpha_M\}| \leq m-1$.
Hence, $|\{h\in N:\tau(R_h')<\alpha_M\}| + |\{r\in\{0,\ldots,n\}:\alpha_{S_r}<\alpha_M\}|
\leq n-1$.
Thus, $l'\geq\alpha_M$.

For the upper endpoint, if $r\geq m+1$, then $S_r\supsetneq M$.
By the choice of $M$, $\alpha_{S_r}=\beta_{S_r}$.
Moreover, by monotonicity, $\alpha_{S_r}\leq\alpha_M<z$.
Hence, $\beta_{S_r} < z$ for every $r\geq m+1$.
Therefore, $|\{h\in N:\tau(R_h')\leq z\}| + |\{r\in\{0,\ldots,n\}:\beta_{S_r}\leq z\}| \geq (m+1)+(n-m) =n+1$.
On the other hand, for every $r\leq m$, since $S_r\subseteq M$, monotonicity implies $\beta_{S_r}\geq\beta_M>z$.
Hence, $|\{h\in N:\tau(R_h')\geq z\}| + |\{r\in\{0,\ldots,n\}:\beta_{S_r}\geq z\}| \geq (n-m+1)+(m+1) \geq n+1$.
Thus, $u'=z$.

Therefore, $F(R')=[l',z]$ with $l'\geq\alpha_M$.
Since $\tau(R_i)=0$, we have $w_{R_i}(F(R))=\beta_M$ and $w_{R_i}(F(R'))=z$.
Since $z<\beta_M$, single-peakedness implies $w_{R_i}(F(R'))P_iw_{R_i}(F(R))$.
Hence, $F(R')P_i^{P}F(R)$.

Similarly, since $\tau(R_j)=1$, we have $w_{R_j}(F(R))=\alpha_M$ and $w_{R_j}(F(R'))=l'$.
Since $l'\geq\alpha_M$, single-peakedness implies
$w_{R_j}(F(R'))R_jw_{R_j}(F(R))$.
Hence, $F(R')R_j^{P}F(R)$.

Thus, the coalition $\{i,j\}$ has a profitable deviation under the pessimistic extension, contradicting pairwise P-strategy-proofness.
Therefore, no $F\in\mathcal{G}$ satisfying voter sovereignty is pairwise P-strategy-proof.
\end{proof}

\section*{Declaration}

\subsection*{Conflict of Interest}
The author declares  that there are no conflicts of interests.

\subsection*{Funding}
This research did not receive any specific grant from funding agencies in the public, commercial, or not-for-profit sectors.

\subsection*{Data availability}
No data was used for the research described in the article.

\subsection*{Disclosure on the use of AI}
During the preparation of this work, the author used ChatGPT for language editing and to identify possible
mathematical or expository issues.
The author takes full responsibility for the content of the manuscript.

\bibliography{references_target}

\end{document}